\documentclass[prl,twocolumn,english,floatfix,superscriptaddress,nofootinbib,longbibliography,aps]{revtex4-2}

\usepackage[english]{babel}
\usepackage{textcomp,xcolor,natbib}
\usepackage{dcolumn}
\usepackage{graphicx}
\graphicspath{{Figures/}}

\usepackage{amsmath, bbm}
\usepackage{latexsym}
\usepackage{amsfonts}   
\usepackage{amssymb}
\usepackage{array}      
\usepackage{epsfig}
\usepackage{txfonts}
\usepackage{xcolor,braket}
\usepackage[colorlinks=true,linkcolor=blue,urlcolor=blue,citecolor=blue]{hyperref}
\usepackage[normalem]{ulem}
\usepackage{soul}
\usepackage{dsfont}

\usepackage{xfrac}  
\usepackage{relsize}

\usepackage{amsthm}
\usepackage{mathrsfs, mathtools, amsmath}

\usepackage[normalem]{ulem}
\usepackage{cancel}
\usepackage{enumitem}

\usepackage{xspace}
\usepackage{orcidlink}

\newcommand{\ketbra}[1]{{\ket{#1}\bra{#1}}}
\newcommand{\id}{{\mathbbm{1}}}
\newcommand{\abs}[1]{\left\lvert #1 \right\rvert}
\newcommand{\norm}[1]{\left\lVert #1 \right\rVert}
\newcommand{\tr}{{\operatorname{tr}}}

\newcommand{\schro}{Schr\"odinger\xspace}
\newcommand{\heis}{Heisenberg\xspace}

\newtheorem{theorem}{Theorem}
\newtheorem{corollary}{Corollary}
\newtheorem{example}{Example}
\newtheorem{proposition}{Proposition}

\definecolor{lavender}{rgb}{0.75, 0.58, 0.89}

\begin{document}
\title{Revivals of Bell nonlocality require Schr\"odinger and Heisenberg non-Markovianity}

\author{Federico Settimo\, \orcidlink{0000-0002-0123-6950}}
\email{fesett@utu.fi}
\affiliation{Department of Physics and Astronomy,
University of Turku, FI-20014 Turun yliopisto, Finland}

\author{Kimmo Luoma\, \orcidlink{0000-0003-3118-612X}}
\affiliation{Department of Physics and Astronomy,
University of Turku, FI-20014 Turun yliopisto, Finland}

\author{Jyrki Piilo\, \orcidlink{0000-0002-5595-873X}}
\affiliation{Department of Physics and Astronomy,
University of Turku, FI-20014 Turun yliopisto, Finland}

\author{Andrea Smirne\, \orcidlink{0000-0003-4698-9304}}
\affiliation{Dipartimento di Fisica ``Aldo Pontremoli'', Universit{\`a} degli Studi di Milano, Via Celoria 16, I-20133 Milan, Italy}
\affiliation{Istituto Nazionale di Fisica Nucleare, Sezione di Milano, Via Celoria 16, I-20133 Milan, Italy}

\author{Bassano Vacchini\, \orcidlink{0000-0002-7574-9951}}
\affiliation{Dipartimento di Fisica ``Aldo Pontremoli'', Universit{\`a} degli Studi di Milano, Via Celoria 16, I-20133 Milan, Italy}
\affiliation{Istituto Nazionale di Fisica Nucleare, Sezione di Milano, Via Celoria 16, I-20133 Milan, Italy}

\author{Dariusz Chru\'sci\'nski\, \orcidlink{0000-0002-6582-6730}}
\affiliation{Institute of Physics, Faculty of Physics, Astronomy and Informatics,
Nicolaus Copernicus University, Grudziadzka 5/7, 87-100 Toru\'{n},
Poland}

\begin{abstract}
    Bell nonlocality is a key resource in quantum information, demonstrating the nonclassicality of quantum theory.
    Noise, however, is in general detrimental to nonlocality, and can cause the loss of the ability to violate any Bell inequality.
    Memory effects, on the other hand, can restore this quantumness and, as recently shown, they can be differently characterized in the Schr\"odinger and in the Heisenberg picture.
    Here, we show that if memory effects allow for revivals in time of nonlocality, then the dynamics must be non-Markovian in both pictures.
    We showcase our findings through a device-independent quantum key distribution task, for which Bell nonlocality is necessary.
\end{abstract}

\maketitle


A striking feature of quantum mechanics is that spatially separated
systems can exhibit correlations that defy any classical explanation.
Such nonlocal correlations can be witnessed by the violation of a Bell
inequality \cite{Bell1964, Clauser1969, Bell2004}.
This feature jointly builds on two peculiar quantum
concepts, that do not have any classical counterpart.  First, quantum
systems can be entangled: the global state that cannot be obtained from the local ones, even allowing for classical correlations.
Second, measurements on a quantum system do not need to be compatible: there exist pairs of measurements whose statistics cannot be obtained by marginalizing over a third larger measurement \cite{Ghne2023}.
As it is well known, both entanglement and incompatibility are required to violate any Bell inequality \cite{Fine1982}.

When dealing with realistic systems, interactions with the surrounding environment will be unavoidably present \cite{Breuer-Petruccione, Vacchini-OQS}.
Both these uncontrolled interactions and any local operation or classical communication will diminish both entanglement \cite{Bennett1996, Horodecki2009} and incompatibility \cite{Heinosaari2015-im}, causing the loss of the ability to violate any Bell inequality.
Nevertheless, open system dynamics can present memory effects, known as non-Markovianity \cite{BLPV-colloquium}.
In a recent work, we demonstrated that such memory effects are differently formulated for the \schro or the \heis evolutions \cite{Settimo-SchroHeis}.
{Non-Markovianity in both pictures is ubiquitous: many paradigmatic noise models, such as dephasing, depolarizing, phase covariant dynamics, as well as dynamics obtained by means of stochastic resetting \cite{Carollo2026}, can display noise in both pictures simultaneously.}

{Crucially for this work, non-Markovianity is tightly related to both entanglement and incompatibility.
\schro Markovian dynamics will cause the entanglement to decrease, while non-Markovianity can allow for revivals \cite{RHP}.
Similarly, for \heis Markovian dynamics incompatibility diminishes, with revivals being possible only for \heis non-Markovian dynamics \cite{Settimo-SchroHeis}.
In this work we show that for revivals of other key quantities, like Bell violations or the rate at which device-independent quantum key distribution (DIQKD) can be performed, one needs instead non-Markovianity in both pictures at the same time.}

Beyond its fundamental interest, the possibility of violating any Bell inequality allows for one of the most successful tasks in quantum information, namely, quantum key distribution:
two parties can perform quantum cryptography in a secure way by sharing a secret key that allows them to send encrypted messages \cite{Ekert1992}.
In order to perform this task, the two parties do not need to trust their local devices: it is sufficient that they obey the laws of quantum mechanics \cite{Acin2007, Scarani2009, Vazirani2014}.
This protocol is known as {DIQKD}.
For its enforcement, the two parties must be able to violate a Bell inequality, thereby concluding that the protocol is robust against eavesdropping. The feasibility and effectiveness of this strategy has already been experimentally validated \cite{Nadlinger2022, Liu2022, Zhang2022}.

{Throughout this work, we consider the exact reduced dynamics of the system, without resorting to perturbative approximations.
Accordingly, the notions of non-Markovianity discussed here refer to the exact dynamical map generated by the underlying system–environment evolution.
Our framework assumes that dynamical maps are completely positive and trace preserving (CPTP), while some approximations such as Redfield produce maps that are CPTP only for sufficiently large time; after such time our divisibility framework applies for such approximated descriptions as well.
Furthermore, we assume that the dynamical map is invertible at all times, however our results generalize in a straightforward way to non-invertible dynamics by considering generalized inverses \cite{GI} when the regular inverse does not exist.}


In this Letter, we show that revivals of either entanglement or incompatibility alone 
are not sufficient: to recover or increase Bell nonlocality, the dynamics needs to be non-Markovian both in the \schro and in the \heis picture.
We further provide explicit examples in which Markovianity of the dynamics in at least one picture necessarily implies a decrease in time of nonlocality.
We finally demonstrate that to restore the effectiveness of DIQKD by increasing the transmitted key rate, the dynamics must be non-Markovian in both pictures.
In a companion paper \cite{SchroHeisTasks}, we investigate the roles of \schro and \heis non-Markovianity in different relevant quantum information tasks, including channel capacity and channel discrimination.


{\it Bell nonlocality.}
We consider a scenario in which two parties (Alice and Bob) share an entangled state $\rho$ and their task is to violate a Bell-CHSH inequality \cite{Clauser1969}
\begin{equation}
    \label{eq:Bell_ineq}
    S_{A,A^\prime,B,B^\prime}(\rho) = \abs{\braket{AB}_\rho + \braket{A B^\prime}_\rho + \braket{A^\prime B}_\rho - \braket{A^\prime B^\prime}_\rho}{\leqslant} 2,
\end{equation}
where $\braket{AB}_\rho = \tr[(A\otimes B)\,\rho]$ and $A$, $A^\prime$ ($B$, $B^\prime$) represent {binary} measurements performed by Alice (Bob), having outcomes $a,a^\prime\, (b,b^\prime)\in\{-1,+1\}$.
{In the following, we assume that Alice's half of the entangled state undergoes some noisy evolution, described by a {CPTP} map $\Phi$, while Bob's half is preserved.
The restriction of noise acting only on Alice's side is done for simplicity, and our} results also generalize in a straightforward way to multiple parties scenarios \cite{Plavala2024} or generalized Bell inequalities with more than two settings \cite{Collins2004}.
{Due to noise, the shared state is then given by} $\rho^\prime = (\Phi\otimes\operatorname{id})[\rho]$.
Alternatively, one can describe the {{dynamics}} in the Heisenberg picture, 
in which the evolution of Alice's observables is described by the map $\Phi^*$, namely the completely positive and unital (CPU) adjoint map of $\Phi$ \cite{heinosaari-ziman}.


If one is able to violate at least one Bell inequality using the state $\rho$, then this state is said to be {\it nonlocal}; otherwise, it is said to be local \cite{Brunner2014}.
In this work, we are interested in the properties of the channel $\Phi$, irrespective of the particular initial state prepared.
Hence, we are interested in whether the channel $\Phi$ is {\it nonlocality breaking} (NLB), which means that $(\Phi\otimes\operatorname{id})[\rho]$ does {not} violate any Bell inequality, whatever $\rho$ \cite{Pal2015,Heinosaari2015-ibc}.
Otherwise, $\Phi$ is said to be {\it nonlocality preserving} (NLP).
Alternatively, we have that $\Phi$ is NLB if and only if 
\begin{equation}
    \label{eq:optimal_Bell}
    S(\Phi) = {\max_\rho\max_{A,A^\prime,B,B^\prime}} S_{\Phi^*[A], \Phi^*[A^\prime],B,B^\prime}(\rho){\leqslant} 2.
\end{equation}
It is well known that in the absence of noise Alice and Bob can obtain the maximal violation of the inequality, namely $S(\operatorname{id}) =2\sqrt2>2$ \cite{tsirelson-bound}.

In order to be able to violate any Bell inequality \eqref{eq:Bell_ineq}, the two parties must share an entangled state and be able to perform incompatible measurements \cite{Fine1982}.
In particular, for any pair of incompatible measurements, it is always possible to find an entangled state that allows for the violation of a Bell inequality \cite{Wolf2009}.
The converse is not true, since there exist non-separable states that do not allow for the violation of any Bell inequality \cite{Werner1989}.
It is possible to quantify the amount of incompatibility via the so-called incompatibility monotones, i.e. functions $I$ of POVMs such that $I(A,A^\prime)=0$ if and only if $A$ and $A^\prime$ are compatible and contractive under CPU maps \cite{Heinosaari2015-im}.
{In the special case of unbiased binary POVMs acting on qubits, $I(A, A^\prime)=0$ is equivalent to \cite{Busch1986, Grinko2025}
\begin{equation}
    \label{eq:incop_qubit}
    \norm{\mathbf a+\mathbf a^\prime}+\norm{\mathbf a-\mathbf a^\prime}\geqslant 2,
\end{equation}
where $\mathbf a$ and $\mathbf a^\prime$ are the Bloch vectors associated to the operators $A$ and $A^\prime$ respectively.}
{A CPU map} $\Phi^*$ is said to be {\it incompatibility breaking} (IB) if \cite{Heinosaari2015-ibc}
\begin{equation}
    \label{eq:incompatibility_breaking}
    I\big(\Phi^*[A],\Phi^*[A^\prime]\big)=0\qquad\forall A,A^\prime.
\end{equation}
In analogy to Eq.~\eqref{eq:incompatibility_breaking}, a channel is said to be {\it entanglement breaking} (EB) if for any bipartite state $\rho$, $(\Phi\otimes\operatorname{id})[\rho]$ is separable \cite{Horodecki2003}.
{Both IB and EB channels have been widely studied in the literature under different noise models \cite{Horodecki2003, Heinosaari2015-ibc, heinosaari_breaking_2015}.}

Since for any pair of incompatible measurements it is possible to find a state that violates a Bell inequality \cite{Wolf2009}, 
and since $\Phi$ being EB implies that $\Phi^*$ is IB \cite{Heinosaari2015-ibc}, considering binary measurements we have \cite{Kumari2023}:
\begin{proposition}
    A channel $\Phi$ is NLP if and only if $\Phi^*$ is not IB.
\end{proposition}
As a consequence, if Alice and Bob are able to prepare arbitrary states and perform arbitrary measurements, they can violate Bell inequalities if and only if the channel $\Phi^*$ preserves incompatibility.
IB channels have been widely studied in the literature in connection to the ability to perform Einstein-Podolsky-Rosen steering \cite{Heinosaari2015-ibc}.
This proposition provides an alternative characterization of IB channels as also being NLB for the special case of binary measurements.

{\it \schro and \heis divisibility.}
We now consider a dynamical {map} acting locally only on Alice's side, so that the noise at time $t$ will be described by a CPTP dynamical map $\Phi_t$.
In the \schro picture, it is possible to describe the channel at later times from the channel at earlier times $s\leqslant t$ via the propagator $\Phi^S_{t,s}$ obeying \cite{BLPV-colloquium, rivas-quantum-nm, Chruscinski2022}
\begin{equation}
    \label{eq:divisibility_S}
    \Phi_t = \Phi^S_{t,s}\circ\Phi_s,\qquad 
\end{equation}
so that for any state $\rho$ it holds $\rho(t) = \Phi^S_{t,s}[\rho(s)]$.
In the \heis picture, instead, the forward propagator reads \cite{Settimo-SchroHeis}
\begin{equation}
    \label{eq:divisibility_H}
    \Phi_t^* = {\Phi^H_{t,s}}^*\circ\Phi_s^*,
\end{equation}
and for any observable $X$ it holds $X(t) = {\Phi^H_{t,s}}^*[X(s)]$.

Assuming that the dynamics is invertible at time $s$, the propagators can be written as
\begin{equation}
    \Phi_{t,s}^S = \Phi_t\circ\Phi_s^{-1},\qquad {\Phi_{t,s}^H}^* = \Phi_t^*\circ \left(\Phi_s^*\right)^{-1}
\end{equation}
and are related by
\begin{equation}
    \label{eq:connecting_props}
    \Phi_{t,s}^H = \Phi_s^{-1}\circ \Phi_{t,s}^S\circ \Phi_s.
\end{equation}
Although both $\Phi_s$ and $\Phi_t$ are CP maps, the propagators $\Phi^S_{t,s}$ and ${\Phi_{t,s}^H}^*$ do not need to be.
In particular, if $\Phi^S_{t,s}$ is CP, the dynamics is said to be \schro { CP-divisible}, while if ${\Phi_{t,s}^H}^*$ is CP, the dynamics is said to be \heis CP-divisible.
Otherwise, if $\Phi^S_{t,s}$ (${\Phi_{t,s}^H}^*$) is not CP, the dynamics is said to be \schro (\heis) non-Markovian.
Violations of CP-divisibility in either picture have been connected to memory effects being present in the dynamics \cite{BLP, BLP-PRA, RHP, Settimo-SchroHeis}
{. Notice that non-Markovianity is not equivalent in the two pictures, since} there exist dynamics that are non-Markovian only in one picture \cite{Settimo-SchroHeis, Carollo2026}.

Furthermore, violations of positivity of the propagators can be witnessed via revivals of suitable norms.
In particular, $\Phi_{t,s}^S$ is positive if and only if \cite{Kossakowski-necessary, Wimann2012}
\begin{equation}
    \label{eq:contractivity_TD}
    \norm{\Phi_t[X]}_1\leqslant\norm{\Phi_s[X]}_1\qquad \forall X,
\end{equation}
with $\norm{X}_1 = \tr\abs X$.
From the point of view of Bell {nonlocality}, non-Markovianity is necessary for revivals in time of the two quantities ruling the necessary conditions for {violations of the} inequality.
For instance, if  there exists a bipartite state $\rho$ such {that} the entanglement at time {$t$} is larger than its entanglement at time {$s$} 
\begin{equation}
    \label{eq:revival_ent}
    E\big((\Phi_t\otimes\operatorname{id})[\rho]\big) > E\big((\Phi_s\otimes\operatorname{id})[\rho]\big),
\end{equation}
then one can conclude that the dynamics is \schro non-Markovian \cite{RHP}.
{As a special case, if $\Phi_s$ is EB and $\Phi_t$ is not, then the dynamics must be \schro non-Markovian.}
Here, $E$ is {an} entanglement monotone \cite{Horodecki2009} and in the following we fix it to be the concurrence \cite{Hill1997, Wootters1998}.
On the other hand, if revivals in time of incompatibility are present for some POVMs $A$, $A^\prime$, {namely}
\begin{equation}
    \label{eq:revival_incop}
    I\big(\Phi_t^*[A], \Phi_t^*[A^\prime]\big) > I\big(\Phi_s^*[A], \Phi_s^*[A^\prime]\big),
\end{equation}
then the dynamics must be \heis non-Markovian \cite{Settimo-SchroHeis}.
{In particular, if $\Phi_s^*$ is EB and $\Phi_t^*$ is not, then the dynamics must be \heis non-Markovian.}
Therefore, if the {dynamics} is CP-divisible in both pictures, Bell nonlocality will diminish over time.
Notice that, due to the inequivalence of non-Markovianity in the two pictures, there exist dynamics presenting revivals of incompatibility but not of entanglement, and viceversa.
{We stress that both the IB property and CP divisibility can be tested by the Choi state of either the map $\Phi_t$ or the propagator in either picture \cite{RHP, SchroHeisTasks}.}


{\it Necessary condition for revival of nonlocality.}
In order to violate Bell inequality \eqref{eq:Bell_ineq}, the two parties need to share an entangled state and perform incompatible measurements on it.
We now proceed to show that, even if \schro (\heis) non-Markovianity can increase entanglement (incompatibility), each non-Markovianity separately is not sufficient to {enhance} nonlocality.

We present the main result of this Letter in the following Theorem.
\begin{theorem}
    \label{th:contractivity_nonlocality}
    Consider a {dynamics} described by a family of CPTP maps $\{\Phi_t\}_t$ and let $t\geqslant s$.
    If either $\Phi^S_{t,s}$ or ${\Phi^H_{t,s}}^*$ is CP, then $S(\Phi_t) \leqslant S(\Phi_s)$.
\end{theorem}
\begin{proof}
    {Assume first that the dynamics is \schro CP-divisible, and therefore $\Phi_t^* = \Phi_s^*\circ {\Phi_{t,s}^S}^*$ for CP ${\Phi_{t,s}^S}^*$.}
    Then
    \begin{equation}
        \begin{split}
            S(\Phi_t) &= {\max_\rho\max_{A,A^\prime,B,B^\prime}} S_{\Phi_s^*\circ {\Phi_{t,s}^S}^*[A], \Phi_s^*\circ {\Phi_{t,s}^S}^*[A^\prime],B,B^\prime}(\rho)\\
            &={\max_\rho\max_{\substack{\tilde A = {\Phi_{t,s}^S}^*[A],\\\tilde A^\prime = {\Phi_{t,s}^S}^*[A^\prime]}}\max_{B,B^\prime}} S_{\Phi_s^*[\tilde A], \Phi^*[\tilde A^\prime],B,B^\prime}(\rho)\\
            &\leqslant {\max_\rho\max_{A, A^\prime, B, B^\prime}} S_{\Phi_s^*[ A], \Phi^*[ A^\prime],B,B^\prime}(\rho) = S(\Phi_s),
        \end{split}
    \end{equation}
    where the inequality uses the fact that $\tilde A$ and $\tilde A^\prime$ are still elements of binary POVMs and the last maximization is over the enlarged space of all such elements.

    {If, instead, the dynamics is \heis CP-divisible, the proof that $S(\Phi_t)\leqslant S(\Phi_s)$ follows similarly from the restriction on the maximisation over bipartite states.}
\end{proof}
As a corollary, if the dynamics up to time $s$ is NLB, nonlocality can be {restored} only by a dynamics which is non-Markovian in both pictures.
{Notice that non-Markovianity in both pictures is only a necessary condition for restoring nonlocality, not a sufficient one.}
\begin{corollary}
    \label{th:revival_nonlocality}
    Suppose that at time $s$ the channel $\Phi_s$ is NLB $S(\Phi_s)\leqslant 2$.
    If nonlocality is restored at the later time $t$, $S(\Phi_t)>2$, then the dynamics must be non-Markovian both in the \schro and in the \heis picture.
\end{corollary}

The results of Theorem \ref{th:contractivity_nonlocality} allows us to define a witness for non-Markovianity in both pictures as
\begin{equation}
    \mathcal W(\Phi) = \int_{\dot{S}   (\Phi_t)>0}dt\,\dot{S}(\Phi_t)>0.
\end{equation}
By definition, ${\mathcal W(\Phi)} >0$ requires non-Markovianity in both pictures, while it equals zero whenever at least one of the two propagators $\Phi_{t,s}^S$ or ${\Phi_{t,s}^H}^*$ is CP for all times $s\leqslant t$.
This quantifier is introduced in analogy to \cite{BLP, BLP-PRA}, but refers to failure of divisibility in both \schro and \heis pictures.

{\it Examples.}
We now provide three examples to {highlight} the relevance of the Theorem.
First, we consider a {dynamics} that is CP-divisible in the \schro picture but non-Markovian in the \heis picture, presenting revivals in incompatibility but neither in entanglement nor in nonlocality.
We will then present the complementary scenario, with only \schro non-Markovianity, without revivals of incompatibility nor nonlocality.
Finally we provide an example of dynamics which is non-Markovian in both pictures, so that all quantities present revivals.

\begin{example}[\heis non-Markovian]
    \label{ex:Heis_nM}
We consider a dynamics {of the form of convex mixture of unitaries, i.e.}
  \begin{equation}
        \label{eq:Heis_nM}
        \begin{split}
            \Phi_t[\rho] =& p_1(t)\rho + p_2(t)\,\sigma_z\rho\sigma_z + p_3(t)\,U\rho U^\dagger\\
            &+ p_4(t)\, V\rho V^\dagger,
        \end{split}
    \end{equation}
    for some probability distribution $p_i(t)$ and  $V=U\sigma_z$, with $U=e^{-i \pi \sigma_y/4}$ a rotation around the $y$ axis.
    Assume that, at time $s$, the dynamics is such that $p_1(s) = 1-p_2(s) = (1+\delta)/2$ and $p_3(s) = p_4(s) =0$, so that the resulting dynamics is of dephasing type, with dephasing parameter $\delta$, with $0<\delta\leqslant 1$.
    At time $t$, instead $p_3(t) = 1-p_4(t) = (1+\delta)/2$ and $p_1(t) = p_2(s) =0$.
    The \schro picture propagator reads $\Phi_{t,s}^S[\rho] = U \rho U^\dagger$ and is clearly CP.
    In order to show that the dynamics is \heis non-Markovian, it suffices to find a pair of POVMs such that inequality \eqref{eq:revival_incop} {holds}.
    Let $A=\big\{\ketbra{+_x},\ketbra{-_x}\big\}$ and $A^\prime=\big\{\ketbra{+_y},\ketbra{-_y}\big\}$, where $\ket{\pm_{x,y}}$ are the eigenstates of $\sigma_x$ and $\sigma_y$ respectively.
    From Eq.~\eqref{eq:incop_qubit}, for any $d\leqslant1/\sqrt2$ it holds $I\big(\Phi_s^*(A),\Phi_s^*(A^\prime)\big)=0$.
    On the other hand, $\Phi_t^*(A)$ and $\Phi_t^*(A^\prime)$ are incompatible, giving
    \begin{equation}
        I\big(\Phi_t^*(A),\Phi_t^*(A^\prime)\big)>I\big(\Phi_s^*(A),\Phi_s^*(A^\prime)\big)=0
    \end{equation}
    which allows us to conclude that the {dynamics} is \heis non-Markovian.
    However, we have
    \begin{equation}
        \max_{A,A^\prime}I\big(\Phi_t^*(A),\Phi_t^*(A^\prime)\big) = \max_{A,A^\prime}I\big(\Phi_s^*(A),\Phi_s^*(A^\prime)\big)
    \end{equation}
 so that  $S(\Phi_{t}) \leqslant S(\Phi_{s})$  and there is no nonlocality revival.
\end{example}





\begin{example}[\schro non-Markovian]
    Similarly to Eq.~\eqref{eq:Heis_nM} we consider
    \begin{equation}
        \label{eq:Schro_nM}
        \begin{split}
            \Phi_t[\rho] =& p_1(t)\rho + p_2(t)\,\sigma_z\rho\sigma_z + p_3(t)\,U\rho U^\dagger\\
            &+ p_4(t)\, W\rho W^\dagger ,
        \end{split}
    \end{equation}
    with the previously considered probability distribution but now $W=\sigma_z U$.
    Then, one has ${\Phi_{t,s}^H}^*[X] = U^\dagger X U$ and therefore the dynamics is \heis Markovian.
    Nevertheless, it is \schro non-Markovian for any $\delta\ne1$, which follows from the violation of Eq.~\eqref{eq:contractivity_TD}
    \begin{equation}
        1=\norm{\Phi_t[\sigma_x]}_1 > \norm{\Phi_s[\sigma_x]}_1 = \delta.
    \end{equation}
On the other hand, due to CP of ${\Phi_{t,s}^H}^*$, the dynamics does not present any revival of nonlocality $S(\Phi_s) = S(\Phi_t) = 0$.
\end{example}

\begin{example}[Non-Markovian in both pictures]
    \label{ex:depolarizing}
    {We now analyze a depolarizing noise of the form
    \begin{equation}
        \label{eq:depolarizing}
        \Phi_t[\rho] = p(t)\,\rho + (1-p(t))\frac{\id}2\tr[\rho],\qquad 0 \leqslant p(t) \leqslant 1.
    \end{equation}
    The dynamics is IB for $p(t)\leqslant 2/3$ and also EB for $p(t)\leqslant 1/3$ \cite{Heinosaari2015-ibc}.
    Since $S(\Phi_t) = p(t)\,2\sqrt{2}$, one has $S(\Phi_t)>S(\Phi_s)$ whenever $p(t)>p(s)$, which is indeed the condition for non-Markovianity of the dynamics and, according to Theorem \ref{th:contractivity_nonlocality}, both $\Phi_{t,s}^S$ and ${\Phi_{t,s}^H}^*$ must be non CP.
    In particular, if $p(s) \leqslant 1/3$ so that $\Phi_s$ is NLB, then from Corollary \ref{th:revival_nonlocality} the non-Markovianity in both pictures allows to restore nonlocality provided $p(t) > 1/3$.

    Alternatively, one can conclude that the dynamics is non-Markovian in both pictures 
    exploiting commutativity of the dynamics, that is  $[\Phi_s,\Phi_t] = 0$, so that according to Eq.~\eqref{eq:connecting_props} we have $\Phi_{t,s}^S={\Phi_{t,s}^H}^*$.}
\end{example}

{\it DIQKD.}
Beyond their fundamental interest, our results impose strict constraints on DIQKD.
Suppose that Alice can perform one out of three measurements $A_0,A_1,A_2$ and Bob one between  $B_1$ and $B_2$. Then DIQKD ensures that they can share a secret key without the need of trusting their devices \cite{Acin2007, Scarani2009, Vazirani2014}.
The only requirement is that they obey the laws of quantum mechanics and can violate a Bell inequality.
{If the two parties choose the optimal measurement settings and initial entangled state, then the possibility of performing DIQKD under the presence of noise requires $S(\Phi)>2$, as in Eq.~\eqref{eq:optimal_Bell}.
In other words, DIQKD is not possible under NLB noise.}

Due to the noise, which can represent either an eavesdropper or simply unwanted environmental interactions, the transmitted key rate $r$ is bounded by the Devetak-Winter bound \cite{Devetak2005, Acin2007}
\begin{equation}
    \label{eq:key_rate}
    r_{DW}(\Phi) =  1- h\big(Q(\Phi)\big)-h\left(\frac{1+\sqrt{\big(S(\Phi)/2\big)^2-1}}{2}\right).
\end{equation}
Here, $h$ is the binary entropy, $Q$ is the quantum bit error rate $Q=\operatorname{prob}(a_0\ne b_1)$, and $S(\Phi)$ is the optimal Bell parameter of Eq.~\eqref{eq:optimal_Bell} obtained from the measurement settings $A_{1,2}$ and $B_{1,2}$.
If Alice and Bob cannot violate any Bell inequality, then they cannot perform DIQKD.

From Theorem \ref{th:contractivity_nonlocality} {and the fact that $r_{DW}$ is a monotonic function of $S(\Phi)$}, it follows that revivals in the bound of the QKD rate require non-Markovianity in both pictures.
Notice that if the two parties optimize their measurement choices, then $Q$ is also monotonic under \schro or \heis Markovian evolutions.
Our Theorem then immediately implies the following result.
\begin{corollary}
    \label{corr:DIQKD}
Consider a dynamics described by a family of CPTP maps $\{\Phi_t\}_t$ and let $t\geqslant s$.
    If $r_{DW}(\Phi_t) > r_{DW}(\Phi_s)$, then the dynamics is non-Markovian in both pictures.
\end{corollary}

\begin{figure}[!t]
    \centering
    \includegraphics[width=\linewidth]{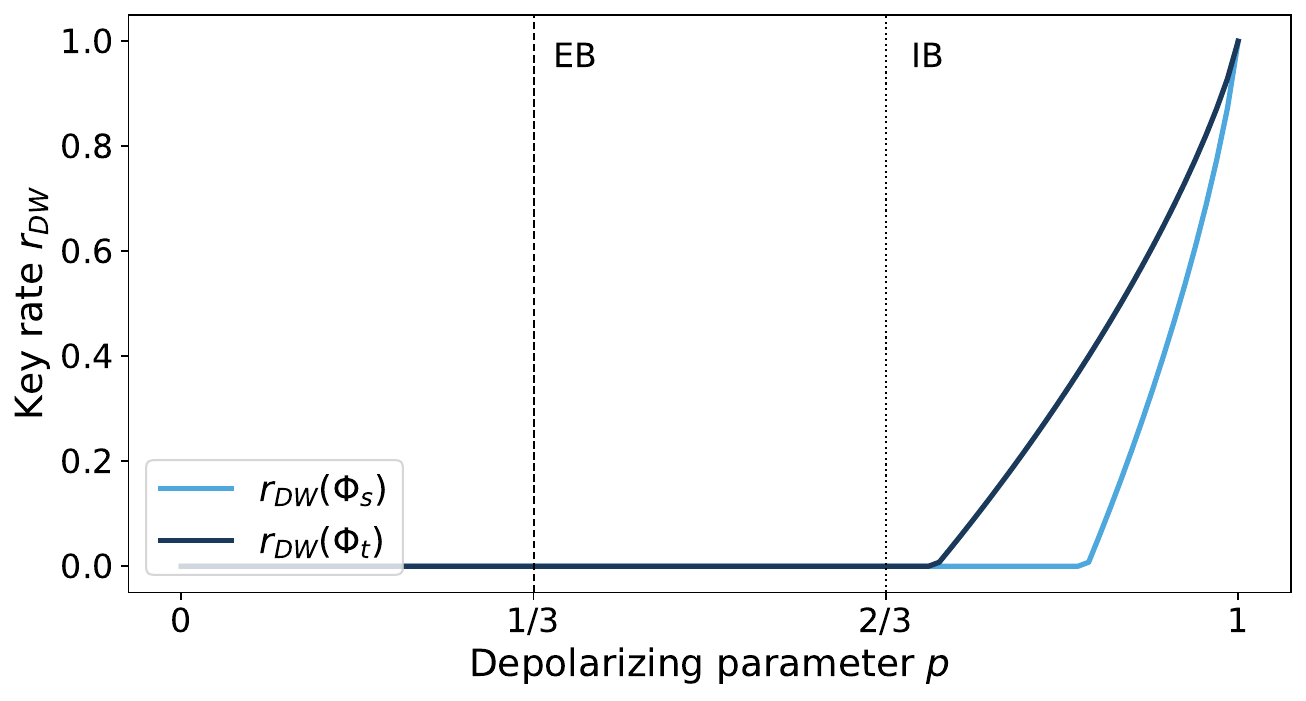}
    \caption{Key rate for the DIQKD of Example \ref{ex:QKD} as a function of the depolarizing parameter $p$ of Eq.~\eqref{eq:depolarizing}, both at time $s$ (light blue) and $t$ (dark blue), as in Eqs.~\eqref{eq:key_rate_dep} and \eqref{eq:key_rate_dep_fin}.
Whenever $r_{DW}(\Phi_t)>r_{DW}(\Phi_s)$ one can conclude that the dynamics is non-Markovian both in the \schro and in the \heis picture, as in Corollary \ref{corr:DIQKD}.
    The dashed (dotted) vertical line is the value of $p$ for which the $\Phi_s$ becomes EB (IB).
    }
    \label{fig:QKD}
\end{figure}

\begin{example}[DIQKD]
    \label{ex:QKD}
    {Consider the same dynamics of Example \ref{ex:depolarizing}, in which $\Phi_t$ is a depolarizing channel.}
    Then, for the optimal measurements setting $A_0=B_1=\sigma_z$, $A_1=(\sigma_z+\sigma_x)/\sqrt{2}$, $A_2=(\sigma_z-\sigma_x)/\sqrt{2}$, $B_2=\sigma_x$ it holds that 
    $S(\Phi_t) = 2\sqrt{2}\big(1-2Q(\Phi_t)\big)$ \cite{Acin2007}.
    {For the sake of definitiveness, we let the depolarizing parameter at time $s$ be $p(s)=p\in[0,1]$ and fix $p(t) = (1+p)/2>p(s)$, so that the dynamics is non-Markovian in both pictures and the depolarizing noise is partially removed.}
    At time $s$, the key rate bound is a monotonically increasing function of the depolarizing parameter $p$ 
    \begin{equation}
        \label{eq:key_rate_dep}
        r_{DW}(\Phi_s) = 1-h\left(\frac{1-p}2\right)-h\left(\frac{1+{\sqrt{2 p^2-1}}}2\right).
    \end{equation}
    At time $t$ it holds that
    \begin{equation}
        \label{eq:key_rate_dep_fin}
        r_{DW}(\Phi_t) = 1-h\left(\frac{1-p}4\right)-h\left({\frac{1+\sqrt{(1+p)^2/2-1}}2}\right)>r_{DW}(\Phi_s)
    \end{equation}
    whenever $p\ne1$, which is equivalent to non-Markovianity in both pictures.
    {The obtained key rates are presented in Fig.~\ref{fig:QKD} as a function of the depolarizing parameter $p$.
    {Notice that the dynamics is non-Markovian in both pictures for any $p\ne1$, however revivals in nonlocality and performance of DIQKD are present only for some values of $p$.}
    }
\end{example}

In this example, the optimal measurements $A_i$ and $B_i$ as well as state $\rho$ are the same at any time.
In general, instead, one might have to make  different choices at different times, however the same considerations hold in a straightforward way.


{\it Conclusions and outlook.} 
In this Letter, we have established a direct connection between revivals over time of Bell nonlocality and the recently introduced distinction between \schro and \heis non-Markovianity.
We showed that in order to restore the nonlocality lost at an intermediate time the dynamics must be non-Markovian both in the \schro and in the \heis picture.

Indeed, while \schro non-Markovianity can restore entanglement and \heis non-Markovianity can restore incompatibility, Bell nonlocality can present revivals only when both mechanisms are present together.
We further showed that the same conclusion applies to DIQKD, for which any revival in the Devetak-Winter bound for the key rate requires non-Markovianity in both pictures.
In this sense, our work provide an operational task that shows the intrinsic relevance of memory in both pictures.

Our work opens several directions for future research.
An important question concerns the investigation of relevant tasks when considering the extension of these results to multipartite Bell scenarios, and higher-dimensional systems or generalized measurement settings.
More broadly, our results further support the viewpoint that \schro and \heis non-Markovianity capture genuinely different physical aspects of memory effects in open quantum systems.

\section*{Acknowledgments}
FS acknowledges support from Magnus Ehrnroothin S\"a\"ati\"o.
DC was supported by the Polish National Science Center under Projects No. 2024/55/B/ST2/01781.
The authors thank the Toru\'n group and the Aleksander Jab\l o\'nski Foundation for hospitality received.

\bibliography{biblio}

\end{document}